\documentclass[american,aps,pra,reprint, superscriptaddress]{revtex4-1}
\usepackage{amsmath,amscd,amsthm,amssymb}
\usepackage{mathrsfs}
\usepackage{graphicx}
\usepackage{epsf,epstopdf}
\usepackage[center]{caption2}
\usepackage[all]{xy}

\usepackage[unicode=true,pdfusetitle, bookmarks=true,bookmarksnumbered=false,bookmarksopen=false, breaklinks=false,pdfborder={0 0 0},backref=false,colorlinks=false] {hyperref}
\hypersetup{ colorlinks,linkcolor=myurlcolor,citecolor=myurlcolor,urlcolor=myurlcolor}
\usepackage{graphics,epstopdf,graphicx, amsthm, amsmath, amssymb, times, braket, color, bm}
\usepackage[up]{subfigure}
\usepackage{cleveref}

\definecolor{myurlcolor}{rgb}{0,0,0.7}

\theoremstyle{plain}
\newtheorem{thm}{\protect\theoremname}
\newtheorem{prop}[thm]{Proposition}

\providecommand{\theoremname}{Theorem}
\newcommand*{\myproofname}{Proof}

\makeatother

\theoremstyle{definition}

\theoremstyle{remark}
\newtheorem{rem}{Remark}[thm]
\newtheorem*{note}{Note}

\begin{document}

 \author{Chunhe Xiong}
 \email{xiongchunhe@zju.edu.cn}
 \affiliation{School of Mathematical Sciences, Zhejiang University, Hangzhou 310027, PR~China}

\author{Asutosh Kumar}
 \email{Corresponding author: asutoshk.phys@gmail.com}
 \affiliation{P.G. Department of Physics, Gaya College, Magadh University, Rampur, Gaya 823001, India}

 \author{Junde Wu}
 \email{Corresponding author: wjd@zju.edu.cn}
 \affiliation{School of Mathematical Sciences, Zhejiang University, Hangzhou 310027, PR~China}

\title{Family of coherence measures and duality between quantum coherence and path distinguishability}
\begin{abstract}

Coherence measures and their operational interpretations lay the cornerstone of coherence theory. In this paper, we introduce a class of coherence measures with $\alpha$-affinity, say $\alpha$-affinity of coherence for $\alpha \in (0, 1)$. Furthermore, we obtain the analytic formulae for these coherence measures and study their corresponding convex roof extension. We provide an operational interpretation for $1/2$-affinity of coherence by showing that it is equal to the error probability to discrimination a set of pure states with the least square measurement. Employing this relationship
we regain the optimal measurement for equiprobable quantum state discrimination. Moreover, we compare these coherence quantifiers, and establish a complementarity relation between $1/2$-affinity of coherence and path distinguishability for some special cases.

\end{abstract}
\maketitle

\section{Introduction}
\label{sec:intro1}

Quantum coherence  \cite{Baumgratz2014} is one of the fundamental features in quantum mechanics, characterizing the wave-like property for all objects. It is also a necessary condition for entanglement and other quantum correlations which manifests its core position in quantum information theory \cite{Adesso2016BC,Yao2015}. As a key quantum resource, coherence may lead to an operational advantage over classical physics, and its important role in quantum algorithms has been investigated \cite{Matera2016A,Hillery2016A,Anand2017,Shi2017}. Hence, for a given quantum state, it is important to ask the amount of coherence it has and if the quantifier of coherence has any operational meaning?
In Ref. \cite{Baumgratz2014}, authors have established a resource theory of coherence which is a rigorous framework to quantify coherence. In this theory, coherence characterizes
superposition of a quantum state relative to a fixed
orthogonal basis and thereafter a lot of work has been done to enrich this theory \cite{Girolami14,Streltsov2015B,Chitambar2016A,chitambar2016B,streltsov2017A}.
This framework places certain important constraints on the measures of coherence, and different coherence measures may
reflect different physical aspects of the quantum system \cite{Yuan2015,winter2016,Napoli2016,Bu2017A}. Like other resource theories, the resource theory of coherence is composed of ``free states'' and ``free operations''.

Let $\mathcal{H}$ be a finite dimensional Hilbert space with an orthogonal basis $\{\ket{i}\}^d_{i=1}$. Density matrices that are diagonal in this basis are free states and we call them incoherent states as they do not possess any coherence. We label this set of incoherent quantum states by $\mathcal{I}$. That is,
\begin{align}
\mathcal{I}=\{\sigma ~|~ \sigma=\sum^d_{i=1}\lambda_i\ket{i}\bra{i}\}.
\end{align}

Free operations in coherence theory are the completely positive and trace preserving (CPTP) maps which admit an incoherent Kraus representation. That is, there always exists a set of Kraus operators $\{K_i\}$ such that
\begin{align}
\frac{K_i\sigma K^{\dagger}_i}{\mathrm{Tr}K_i\sigma K^{\dagger}_i}\in\mathcal{I},
\end{align}
for each $i$ and any incoherent state $\sigma$. These operations are also called incoherent operations and we label these operations by $\Phi$.

Analogous to the quantification of entanglement \cite{HorodeckiRMP09,Vedral1997,Vedral1998,Plenio2000}, any measure of coherence $C$ should satisfy the following axioms \cite{Baumgratz2014}:

(C1) Faithfulness. $C(\rho)\ge0$ with equality if and only if $\rho$ is incoherent.

(C2) Monotonicity. $C$ does not increase under the action of an incoherent operation, i.e., $C(\Phi(\rho))\le C(\rho)$ for any incoherent operation $\Phi$.

(C3) Strong monotonicity. $C$ does not increase on average under selective incoherent operations, i.e., $\sum_ip_iC(\sigma_i)\le C(\rho)$ with probabilities $p_i=\mathrm{Tr}K_i\rho K^{\dagger}_i$, post-measurement states $\sigma_i=p^{-1}_iK_i\rho K^{\dagger}_i$, and incoherent operators $K_i$.

(C4) Convexity. Nonincreasing under mixing of quantum states, i.e.,$\sum_ip_iC(\rho_i)\ge C(\sum_ip_i\rho_i)$ for any set of states $\{\rho_i\}$ and $p_i\ge0$ with $\sum_ip_i=1$.

Conditions (C1) and (C2) highlight the role of free states and free operations in the coherence theory, i.e., the free states have zero coherence and the free operations cannot increase coherence of any state. (C3) and (C4) are two constraints imposed on coherence measures. Like in entanglement theory, a coherence quantifier which satisfies nonnegativity and (strong) monotonicity is called (strong) coherence monotone. Furthermore, if it also satisfies convexity, we call it convex (strong) coherence monotone.

  It is worth noting that authors in Ref. \cite{X.yu2016B} have provided a simple and interesting condition to replace (C3) and (C4) with the additivity of coherence for block-diagonal states,
 \begin{align}\label{eq1}
 C(p\rho\oplus(1-p)\sigma)=pC(\rho)+(1-p)C(\sigma),
 \end{align}
 for any $p\in[0,1],~\rho\in\mathcal{E}(\mathcal{H}_1),~\sigma\in\mathcal{E}(\mathcal{H}_2)$ and $p\rho\oplus(1-p)\sigma\in\mathcal{E}(\mathcal{H}_1\oplus\mathcal{H}_2)$, where $\mathcal{E}(\mathcal{H})$ denotes the set of density matrices on $\mathcal{H}$.

 They proved that conditions (C1), (C2) and \eqref{eq1} are equivalent to conditions (C1) through (C4). This is surprising because (\ref{eq1}) is operation-independent equality, whereas strong monotonicity and convexity are operation-dependent inequalities. In general, it is relatively easy to check whether a coherence quantifier satisfies (\ref{eq1}) than (C3) and (C4).

In this paper we introduce a class of coherence measures, and attempt to answer the question posed in the beginning, by linking this coherence measure to ambiguous quantum state discrimination (QSD).
QSD, as a fundamental problem in quantum mechanics, has been studied extensively \cite{Helstrom1976,Holevo2011,Holevo2001,HELSTROM1967A,HELSTROM1968,HOLEVO1973337,Yuen1975,Davies1978}. It is not only an important problem of theoretical research, but also plays a key role in quantum communication and quantum cryptography \cite{Phoenix1995,HKLo2011,Bouwmeester2000,Gisin2002RMP,Loepp2006}.

 We briefly review the ambiguous QSD. Suppose there are two persons, Alice and Bob. Alice chooses a state $\rho_i$ from a set of states $\{\rho_i\}^N_{i=1}$ with probability $\eta_i$ and sends it to Bob. Now Bob's job is to determine which state he has received, as accurately as possible. To do this, Bob performs a positive-operator valued measure (POVM) on each $\rho_i$ and declares that the state is $\rho_j$ when the measurement outcome reads $j$. The POVM is a set of positive operators $\{M_i\}$ satisfying $\sum_i M_i=I$. As the probability to get the result $j$ with state $\rho_i$ is $p_{j|i}=\mathrm{Tr}(M_j\rho_i)$, the corresponding maximal success probability is
\begin{align}
P^{opt}_S(\{\rho_i,\eta_i\})=\mathop{\mathrm{max}}_{\{M_i\}}\sum_i\eta_i\mathrm{Tr}(M_i\rho_i),
\end{align}
where the maximization is done over all POVMs. For $N=2$ case, the analytic formula of $P^{opt}_S$ and the optimal measurement are known. However, no solution about optimal probability and measurement is known for general $N>2$ case.

As a suboptimal choice, least square measurement (LSM)  is an alternative to discriminate quantum states \cite{Belavkin1975a,Belavkin1975,Holevo1978,Hausladen1994,Hausladen1996,peres1991,Eldar2001}. In comparison to the optimal measurement, the LSM has several nice properties. First, its construction is relatively simple as it can be determined directly from the given ensemble. Second, it is very close to the optimal measurement when the states to be distinguished are {\it almost orthogonal} \cite{Holevo1978,Spehner2014}. The construction of LSM is as follows.

Given an ensemble $\{\rho_i,\eta_i\}^N_{i=1}$ and denoting $\rho_{out}=\sum_i \eta_i\rho_i$, the least square measurements are \cite{note}
\begin{align}\label{eq2}
M^{lsm}_i=\eta_i\rho_{out}^{-1/2}\rho_i\rho_{out}^{-1/2},i=1,2,...,N.
\end{align}

 As a result, the minimal error probability of this measurement is
\begin{align}\label{eq4}
P^{lsm}_E(\{\rho_i,\eta_i\})=1-\sum_i\eta_i\mathrm{Tr}(M^{lsm}_i\rho_i).
\end{align}

The paper is structured as follows. In Sec. \ref{sec:coherence-affinity2} we introduce $\alpha$-affinity of coherence. We reveal the connection between the $1/2$-affinity of coherence and QSD with least square measurement in Sec. \ref{sec:coherence-lsm3}. Furthermore, we deal with quantum state discrimination with coherence theory in Sec. \ref{sec:qsd-lsm4} and Sec. \ref{sec:qsd-lsm5}. Besides, we establish a duality between $1/2$-affinity of coherence and path distinguahability in Sec. \ref{sec:duality6}, and finally conclude in Sec. \ref{sec:conclusion7} with a summary and outlook.

\section{Quantifying coherence with affinity}
\label{sec:coherence-affinity2}

\subsection{$\alpha$-affinity and $\alpha$-affinity of distance}
Distances in state space are good candidates for quantifying quantum correlations. In this subsection, we introduce a distance using which we can establish a bona fide measure to quantify coherence.
In classical statistical theory \cite{LeCam1986}, affinity is defined as
\begin{align*}
A(f,g)=\sum_x\sqrt{f(x)g(x)},
\end{align*}
where $f$ and $g$ are discrete probability distributions. 
This definition is alike the Bhattacharyya coefficient \cite{Bhattacharyya-measure} between two probability distributions (discrete or continuous) in classical probability theory. 
Classical affinity quantifies the closeness of two probability distributions.
Borrowing the notion from classical statistical theory, Luo and Zhang \cite{Luo2004} have introduced quantum affinity as follows.
Let $\mathcal{H}$ be a $d$-dimensional Hilbert space and $\mathcal{E}(\mathcal{H})$ be the set of density matrix on $\mathcal{H}$. For any $\rho,\sigma\in\mathcal{E}(\mathcal{H})$, quantum affinity is defined as
 \begin{align}\label{eq5}
 A(\rho,\sigma):=\mathrm{Tr}(\sqrt{\rho}\sqrt{\sigma}).
 \end{align}

Quantum affinity, similar to fidelity \cite{Nielsen10}, describes how close two quantum states are. We drop the adjective ``quantum'' in the rest of this paper unless there is any ambiguity.

 The notion of affinity has been extended to $\alpha$-affinity $(0< \alpha <1)$, and is defined as
 \begin{align*}
 A^{(\alpha)}(\rho,\sigma):=\mathrm{Tr}\rho^{\alpha}\sigma^{1-\alpha}.
 \end{align*}

For each $\alpha\in(0,1)$, $A^{(\alpha)}(\rho,\sigma)$ satisfies the following properties:
(1) Boundedness. $A^{(\alpha)}(\rho,\sigma)\in[0,1]$ with $A^{(\alpha)}(\rho,\sigma)=1$ if and only if $\rho=\sigma$.
(2) Monotonicity. $A^{(\alpha)}(\rho,\sigma)\le A^{(\alpha)}(\Phi(\rho),\Phi(\sigma))$ for any CPTP map $\Phi$.
(3) Joint concavity. If $\rho_i,\sigma_i\in\mathcal{E}(\mathcal{H})$ and $p_i\ge0,\sum_i p_i=1$, then $A^{(\alpha)}(\sum_ip_i\rho_i,\sum_ip_i\sigma_i)\ge\sum_ip_iA^{(\alpha)}(\rho_i,\sigma_i)$.
The proof of property (1) is given in Appendix \ref{app3}. See Ref. \cite{Audenaert2015} for the property (2), and property (3) is the result of Lieb's concavity theorem \cite{LIEB1973}.

It is well-known that $\alpha$-affinity plays an important role in quantum hypothesis testing. For the two state discrimination with many identical copies, one has \cite{Audenaert2007,Audenaert2007A}
\begin{align}
-\lim_{N\rightarrow\infty}\frac{1}{N}P^{opt}_{E,N}(\{\rho^{\otimes N}_i,\eta_i\}^2_{i=1})=-\inf_{\alpha\in(0,1)}\{\ln(\mathrm{Tr}\rho^{\alpha}_1\rho^{1-\alpha}_2)\}.\nonumber
\end{align}
This limit defines a function of $\alpha$-affinity, and
\begin{align}
Q(\rho,\sigma):=\min_{\alpha\in(0,1)}A^{(\alpha)}(\rho,\sigma),
\end{align}
is nonlogarithmic version of quantum Chernoff bound (QCB) \cite{Audenaert2007A}.



%


 Moreover, we can see that $\alpha$-affinity is related to $\alpha$-$z$-relative R$\acute{\mathrm{e}}$nyi entropy \cite{Audenaert2015}
 \begin{align}
 S_{\alpha,z}(\rho||\sigma)=\frac{1}{\alpha-1}\ln F_{\alpha,z}(\rho||\sigma), \nonumber
 \end{align}
where
\begin{align}\label{eq10}
F_{\alpha,z}(\rho||\sigma):=\mathrm{Tr}(\sigma^{\frac{1-\alpha}{2z}}\rho^{\frac{\alpha}{z}}\sigma^{\frac{1-\alpha}{2z}})^z,
\end{align}
and
\begin{align}
A^{(\alpha)}(\rho,\sigma)=F_{\alpha,1}(\rho,\sigma).
\end{align}

Note that the family of $\alpha$-$z$-relative R$\acute{\mathrm{e}}$nyi entropies includes relative entropy $S$ and max-relative entropy $S_{max}$ \cite{Audenaert2015}
\begin{align}
S=\lim_{\alpha\rightarrow1}S_{\alpha,\alpha}, ~~S_{max}=\lim_{\alpha\rightarrow\infty}S_{\alpha,\alpha}.\nonumber
\end{align}

It's worth noting that several coherence measures like relative entropy \cite{Baumgratz2014}, geometric coherence \cite{Streltsov2015B} and max-relative entropy \cite{Bu2017A} are related to $\alpha$-$z$-relative R$\acute{\mathrm{e}}$nyi entropy. In the next subsection, we introduce yet another measure of coherence, namely $\alpha$-affinity of coherence which is related to $\alpha$-$z$-relative R$\acute{\mathrm{e}}$nyi entropy.

Based on $\alpha$-affinity, we introduce $\alpha$-affinity of distance as
 \begin{align}\label{eq7}
 d^{(\alpha)}_a(\rho,\sigma):=1- [A^{(\alpha)}(\rho,\sigma)]^{1/\alpha},
 \end{align}
where $\rho,\sigma\in\mathcal{E}(\mathcal{H})$. Obviously, $\alpha$-affinity of distance satisfies the following properties.

(P1) $d^{(\alpha)}_a(\rho,\sigma)\ge 0$ with equality if and only if $\rho=\sigma$.

(P2) $d^{(\alpha)}_a$ is contractive under CPTP maps.

\subsection{Quantifying coherence}

Quantification of entanglement from the geometric point of view began in \cite{Vedral1997,Vedral1998}. Authors in these two papers put forward the scheme to quantify entanglement with the minimal distance between a given quantum state and all the separable states with relative entropy and Bures distance. Later, Luo and Zhang \cite{Luo2004} studied the quantification of entanglement using Hellinger distance. Bures distance and Hellinger distance have been proven to be good choices to quantify quantum discord \cite{spehner2013A,spehner2013B,Roga2016}.

For any $\alpha\in(0,1)$, we define $\alpha$-affinity of coherence as the minimal $\alpha$-affinity of distance over all incoherent states,
\begin{align}
C^{(\alpha)}_a(\rho):=&\min_{\sigma\in\mathcal{I}} d^{(\alpha)}_a(\rho,\sigma) \nonumber \\
=&1-\max_{\sigma\in\mathcal{I}}(\mathrm{Tr}(\rho^{\alpha}\sigma^{1-\alpha}))^{1/\alpha}.
\end{align}

An advantage of $C^{(\alpha)}_a$ over geometric coherence, $C_g(\rho):=1-\max_{\sigma\in\mathcal{I}}(\mathrm{Tr}(\sqrt{\sqrt{\sigma}\rho\sqrt{\sigma}}))^2$ \cite{Streltsov2015B}, is that it is relatively easy to compute.
Let $\sigma=\sum_i\mu_i\ket{i}\bra{i}$ be an incoherent state. Then
\begin{align}\label{eq2}
A^{(\alpha)}(\rho)&\equiv \max_{\sigma\in\mathcal{I}}\mathrm{Tr}(\rho^{\alpha}\sigma^{1-\alpha})\nonumber\\
&=\max_{\mu_i} \left(\sum_i\mu_i^{1-\alpha}\bra{i}\rho^{\alpha}\ket{i}\right)\nonumber\\
&\le \max_{\mu_i} \left(\sum_i\mu_i\right)^{1-\alpha} \left(\sum_i\bra{i}\rho^{\alpha}\ket{i}^{1/\alpha}\right)^{\alpha}\nonumber\\
&=\left(\sum_i\bra{i}\rho^{\alpha}\ket{i}^{1/\alpha}\right)^{\alpha},
\end{align}
where the inequality follows from the H$\ddot{\mathrm{o}}$lder's inequality: $\sum_{i=1}^n |x_iy_i| \le \left(\sum_{i=1}^n |x_i|^p\right)^{1/p} \left(\sum_{i=1}^n |y_i|^q\right)^{1/q}$ for $p,q>1$ with $\frac{1}{p} + \frac{1}{q} = 1$. Here $p=\frac{1}{1-\alpha} >1$ and $q=\frac{1}{\alpha} >1$. Inequality \eqref{eq2} gives an upper bound on $A^{(\alpha)}(\rho)$. This suggests that we can choose suitable
$\{\mu_i\}$'s such that above inequality becomes an equality.
As a result, we obtain the analytic expression for $C^{(\alpha)}_a$ as,
\begin{align}\label{eq11}
C^{(\alpha)}_a(\rho)=1-\sum_i\bra{i}\rho^{\alpha}\ket{i}^{1/\alpha},
\end{align}
and the closest incoherent state which minimizes $C^{(\alpha)}_a(\rho)$ is
\begin{align}\label{eq9}
\sigma_{\rho}=\sum_i\frac{\bra{i}\rho^{\alpha}\ket{i}^{1/\alpha}}{\sum_j\bra{j}\rho^{\alpha}\ket{j}^{1/\alpha}}\ket{i}\bra{i}.
\end{align}

With (P1), (P2) and \eqref{eq1}, we have the following theorem.

 \begin{thm}\label{thm4}
 $\alpha$-affinity of coherence is a coherence measure.
 \end{thm}

 \begin{proof}
First, it is obvious that $C^{(\alpha)}_a(\rho)\ge0$. Since $ d^{(\alpha)}_a(\rho,\sigma)=0$ iff $\rho=\sigma$, one has $C^{(\alpha)}_a(\rho)=0$ if and only if $\rho\in\mathcal{I}$. In addition, since $d^{(\alpha)}_a(\rho,\sigma)$ obeys monotonicity under CPTP maps, we have $C^{(\alpha)}_a(\rho)\ge C^{(\alpha)}_a(\Phi(\rho))$ for any incoherent operation $\Phi$.
Now, instead of (C3) and (C4), we prove that $C^{(\alpha)}_a$ satisfies additivity of coherence for block-diagonal states. We have
\begin{align}
&C^{(\alpha)}_a(p\rho\oplus(1-p)\sigma)\nonumber\\
=&1-\sum_i\bra{i}(p\rho\oplus(1-p)\sigma)^{\alpha}\ket{i}^{1/\alpha}\nonumber\\
=&1-\sum_i\bra{i}(p\rho)^{\alpha} \oplus [(1-p)\sigma]^{\alpha}\ket{i}^{1/\alpha}\nonumber\\
=&p(1-\sum_i\bra{i}\rho^{\alpha}\ket{i}^{1/\alpha})+(1-p)(1-\sum_i\bra{i}\sigma^{\alpha}\ket{i}^{1/\alpha})\nonumber\\
=&pC^{(\alpha)}_a(\rho)+(1-p)C^{(\alpha)}_a(\sigma),\nonumber
\end{align}
Thus, $C^{(\alpha)}_a$ is a coherence measure for each $\alpha\in(0,1)$.
\end{proof}


Similarly, we define quantum Chernoff bound of coherence, $C_{qcb}(\rho)$, and affinity of coherence, $\widetilde{C}_a(\rho)$, respectively as
\begin{align}
C_{qcb}(\rho)&:=\min_{\sigma\in\mathcal{I}}(1-Q^{1/\alpha}(\rho,\sigma))\nonumber\\
&=1-\max_{\sigma\in\mathcal{I}}\min_{\alpha\in(0,1)}(\mathrm{Tr}(\rho^{\alpha}\sigma^{1-\alpha}))^{1/\alpha}\nonumber\\
&=\max_{\alpha\in(0,1)}C^{(\alpha)}_a(\rho),
\label{eq:C-qcb}
\end{align}
and
\begin{align}
\widetilde{C}_a(\rho)&:=\min_{\sigma\in\mathcal{I}}(1-A(\rho,\sigma))=1-\max_{\sigma\in\mathcal{I}}\mathrm{Tr}(\sqrt{\rho}\sqrt{\sigma})\nonumber\\
&=1-\sqrt{\sum_i\bra{i}\sqrt{\rho}\ket{i}^2},
\end{align}
and the closest incoherent state is again $\sigma_{\rho}$ in Eq. \eqref{eq9}.

Note that Eq. \eqref{eq:C-qcb} does not necessarily imply that $C_{qcb}$ is a coherence measure for some $\alpha\in(0,1)$ because $C^{\alpha}_a$ is a coherence measure. This can be argued as follows: for a given $\rho$, let $\alpha'$ be the value of $\alpha$ such that $C^{(\alpha')}_a(\rho) = \max_{\alpha}C^{(\alpha)}_a(\rho)$. Then, $C_{qcb}(\rho)=C^{(\alpha')}_a(\rho) \geq C^{(\alpha')}_a[\Phi(\rho)] \leq \max_{\alpha}C^{(\alpha)}_a[\Phi(\rho)]$, where $\Phi$ is an incoherent operation. Thus, it is not immediately clear that $C_{qcb}$ is a coherence measure.
Next, we can show that $\widetilde{C}_a$ is a convex weak coherence monotone. Following the same lines of the proof of Theorem \ref{thm4}, $\widetilde{C}_a$ satisfies (C1) and (C2). Moreover, convexity of $\widetilde{C}_a$ can be derived from the joint concavity of $A(\rho,\sigma)$. However, $\widetilde{C}_a$ does not satisfy strong monotonicity.
 \begin{align}
 \text{Let}~~\rho_1=\frac{1}{2}\begin{pmatrix}
1&1\\
1&1
\end{pmatrix}~~ \text{and}~~
\rho_2=\frac{1}{3}\begin{pmatrix}
1&1&1\\
1&1&1\\
1&1&1
\end{pmatrix}.\nonumber
\end{align}
Then $\widetilde{C}_a(\rho_1)=1-\sqrt{\frac{1}{2}}$, $\widetilde{C}_a(\rho_2)=1-\sqrt{\frac{1}{3}}$, and
\begin{align}
\widetilde{C}_a \left(\frac{1}{2}\rho_1\oplus\frac{1}{2}\rho_2\right)&=1-\sqrt{\frac{5}{12}}\nonumber\\
&\ne\frac{1}{2}(\widetilde{C}_a(\rho_1)+\widetilde{C}_a(\rho_2))\nonumber.
\end{align}
In conclusion, $\widetilde{C}_a$ is a convex weak coherence monotone.

\subsection{Coherence for pure states and single-qubit states}

In this subsection, we evaluate $\alpha$-affinity of coherence for pure states and single-qubit states.
For any pure state $\ket{\psi}$,
\begin{align}\label{eq18}
C^{(\alpha)}_a(\ket{\psi})=1-\sum_i|\bra{i}\psi\rangle|^{2/\alpha},
\end{align}
is a non-increasing function of $\alpha$. 
We have $C^{\alpha}_a(\ket{\psi})\rightarrow1$ when $\alpha\rightarrow0$. This is very interesting observation that all coherent pure states are almost the maximally coherent states. 
If we consider the convex roof extension of $\alpha$-affinity of coherence for a mixed state $\rho$ as, 
\begin{align}
C^{(\alpha)}_{a}(\rho):=\min_{\{p_i,\ket{\psi_i}\}}\sum_ip_iC^{(\alpha)}_a(\ket{\psi_i}),
\end{align}
then $\mbox{lim}_{\alpha \rightarrow 0} C^{(\alpha)}_{a}(\rho)=1$. That is, $\mbox{lim}_{\alpha \rightarrow 0}C^{(\alpha)}_{a}$ is a coherence measure which equals to unity when the state is coherent and is zero otherwise. 
A similar measure, namely {\it trivial coherence measure}, was discussed in Ref. \cite{trivial-coherence-measure} for which similar consequences were observed.

For a single-qubit state $\rho=\frac{1}{2}(I+\sum_ic_i\sigma_i)$ with $\sigma_i~(i=1,2,3)$ being Pauli matrices, the eigenvalues are $\lambda_{1,2}=(1\mp|{\bf c}|)/2$ and


 \begin{align}
 \rho^{\alpha}=\begin{pmatrix}
\frac{\lambda^{\alpha}_1+\lambda^{\alpha}_2}{2}+\frac{c_3(\lambda^{\alpha}_2-\lambda^{\alpha}_1)}{2|{\bf c}|} & \frac{(-c_1+ic_2)(\lambda^{\alpha}_1-\lambda^{\alpha}_2)}{2|{\bf c}|}\\
\frac{(-c_1-ic_2)(\lambda^{\alpha}_1-\lambda^{\alpha}_2)}{2|{\bf c}|}& \frac{\lambda^{\alpha}_1+\lambda^{\alpha}_2}{2}-\frac{c_3(\lambda^{\alpha}_2-\lambda^{\alpha}_1)}{2|{\bf c}|}
\end{pmatrix}.\nonumber
 \end{align}

Therefore, the corresponding $\alpha$-affinity of coherence is
\begin{align}\label{eq6}
C^{(\alpha)}_a(\rho)=1-(A+B)^{1/\alpha}-(A-B)^{1/\alpha},
\end{align}
where
\begin{align}
A&=\frac{(\frac{1-|{\bf c}|}{2})^{\alpha}+(\frac{1+|{\bf c}|}{2})^{\alpha}}{2},~~\text{and}\nonumber\\
B&=\frac{c_3((\frac{1+|{\bf c}|}{2})^{\alpha}-(\frac{1-|{\bf c}|}{2})^{\alpha})}{2|{\bf c}|}.\nonumber
\end{align}

%
%
%

\section{1/2-affinity of coherence and least square measurement}
\label{sec:coherence-lsm3}

Spehner and Orszag \cite{Spehner2017A} first revealed the connection between quantum correlation (Hellinger distance based quantum discord) and QSD with least square measurement.
In coherence theory, there is a very close relationship between geometric coherence and QSD. Authors in Ref. \cite{Xiong2018A} have recently shown that geometric coherence of $\rho$ is equal to the minimum error probability to discriminate a set of linearly independent pure states $\{\ket{\psi_i},\eta_i\}^d_{i=1}$ with von Neumann measurement, where $\ket{\psi_i}=\eta^{-1/2}_i\sqrt{\rho}\ket{i}$, $\eta_i=\rho_{ii}$ and $d=\mathrm{rank}(\rho)$. Since the optimal measurement is not easy to find, we consider the least square measurement for $\{\ket{\psi_i},\eta_i\}^d_{i=1}$.

For $\{\ket{\psi_i},\eta_i\}^d_{i=1}$, there are two cases.  If $\eta_i\ne0~(i=1,...,d)$, then the ensemble contains $d$ states.
Since $\sum_i\eta_i\ket{\psi_i}\bra{\psi_i}=\rho$, the least square measurement is
\begin{align}\label{eq20}
M^{lsm}_i=\eta_i\rho^{-1/2}\ket{\psi_i}\bra{\psi_i}\rho^{-1/2}=\ket{i}\bra{i},
\end{align}
where $\rho^{-1/2}:=\sum_i\lambda^{-1/2}\ket{a_i}\bra{a_i}$ if $\rho=\sum_i\lambda\ket{a_i}\bra{a_i}$ is the spectral decomposition. Thus, $\sum_iM^{lsm}_i=I$ and the successful probability to discriminate the ensemble $\{\ket{\psi_i},\eta_i\}^d_{i=1}$ with $\{M_i^{lsm}\}$ is
\begin{align}
P^{lsm}_S(\{\ket{\psi_i},\eta_i\}^d_{i=1})&=\sum_i\eta_i\mathrm{Tr}(M^{lsm}_i\ket{\psi_i}\bra{\psi_i})\nonumber\\
&=\sum_i\bra{i}\sqrt{\rho}\ket{i}^2\nonumber\\
&=[A^{(1/2)}(\rho)]^2.
\end{align}

If $\eta_i=0$ for some $i=i_1,i_2,...,i_s$, then
the ensemble  $\{\ket{\psi_i},\eta_i\}^d_{i=1}$ reduces to $\{\ket{\psi_{i^{\prime}}},\eta_{i^{\prime}}\}^{d-s}_{i^{\prime}=1}$. In fact, as $\eta_i=\bra{i}\rho\ket{i}=|\sqrt{\rho}\ket{i}|^2$, $\eta_i=0$ implies $\ket{\psi_i}$ is a zero vector. If $\mathcal{S}$ is the subspace spanned by $\{\ket{\psi_{i^{\prime}}}\}^{d-s}_{i^{\prime}=1}$, then
 \begin{align*}
M^{lsm}_{i^{\prime}}=\eta_{i^{\prime}}\rho^{-1/2}\ket{\psi_{i^{\prime}}}\bra{\psi_{i^{\prime}}}\rho^{-1/2}=\ket{i^{\prime}}\bra{i^{\prime}},
\end{align*}
for all $i^{\prime}$, and $\sum^{d-s}_{i^{\prime}}M^{lsm}_{i^{\prime}}=I_{\mathcal{S}}.$
Moreover, the successful probability to discriminate the ensemble $\{\ket{\psi_{i^{\prime}}},\eta_{i^{\prime}}\}^{d-s}_{i^{\prime}=1}$ with $\{M^{lsm}_{i^{\prime}}\}$ is
\begin{align*}
P^{lsm}_S(\{\ket{\psi_{i^{\prime}}},\eta_{i^{\prime}}\}^{d-s}_{i^{\prime}=1})&=\sum^{d-s}_{i^{\prime}=1}\eta_{i^{\prime}}\mathrm{Tr}(M^{lsm}_{i^{\prime}}\ket{\psi_{i^{\prime}}}\bra{\psi_{i^{\prime}}})\\
&=\sum^{d-s}_{i^{\prime}=1}\bra{i^{\prime}}\sqrt{\rho}\ket{i^{\prime}}^2\\
&=\sum^d_{i=1}\bra{i}\sqrt{\rho}\ket{i}^2\\
&=[A^{(1/2)}(\rho)]^2.
\end{align*}

In other words, the corresponding error probability to discriminate linearly independent pure states $\{\ket{\psi_i},\eta_i\}^d_{i=1}$ is
\begin{align}
P^{lsm}_E(\{\ket{\psi_i},\eta_i\})=1-P^{lsm}_S(\{\ket{\psi_i},\eta_i\}^d_{i=1})=C^{(1/2)}_a(\rho)\nonumber.
\end{align}

Thus, we have the following theorem.

\begin{thm}\label{thm1}
If quantum state $\rho$ describes a quantum system in $d$-dimensional Hilbert space $\mathcal{H}$
with $\{\ket{i}\}^d_{i=1}$ being a reference basis, then the $\alpha$-affinity of coherence of $\rho$ is equal to the error probability to discriminate $\{\ket{\psi_i},\eta_i\}^d_{i=1}$ with least square measurement. That is,
\begin{align}\label{eq8}
C^{(1/2)}_a(\rho)=P^{lsm}_E(\{\ket{\psi_i},\eta_i\}^d_{i=1}),
\end{align}
where $\eta_i=\bra{i}\rho\ket{i}$ and $\ket{\psi_i}=\eta^{-1/2}_i\sqrt{\rho}\ket{i}$.
\end{thm}

\begin{rem}
If $\rho$ is an incoherent state, then $C^{1/2}_a(\rho)=0$ which means that $\{\ket{\psi_i},\eta_i\}_i$  can be perfectly discriminated by the least square measurement. In other words, the LSM is actually the optimal measurement.
\end{rem}

\section{Least square measurement and optimal measurement}
\label{sec:qsd-lsm4}

\subsection{QSD with LSM and $1/2$-affinity of coherence}

In this section, we review a connection between the least square measurement (as a suboptimal choice) and the optimal measurement in QSD protocol. Authors in Ref. \cite{Xiong2018A} have linked quantum state discrimination to geometric coherence.

Let us consider QSD of a set of pure states $\{\ket{\psi_i},\eta_i\}^d_{i=1}$. Denote a matrix $M$ with $M_{ij}=\sqrt{\eta_i\eta_j}\langle\psi_i|\psi_j\rangle$, $1\le i,j\le d$, that is,
 \begin{align}\label{eq10}
 M=\begin{pmatrix}
\eta_1&\sqrt{\eta_1\eta_2}\langle\psi_1|\psi_2\rangle&...& \sqrt{\eta_1\eta_d}\langle\psi_1|\psi_d\rangle\\
\sqrt{\eta_2\eta_1}\langle\psi_2|\psi_1\rangle&\eta_2&...&\sqrt{\eta_2\eta_d}\langle\psi_2|\psi_d\rangle\\
.&.&...&.\\.&.&...&.\\
\sqrt{\eta_d\eta_1}\langle\psi_d|\psi_1\rangle&\sqrt{\eta_d\eta_2}\langle\psi_d|\psi_2\rangle&...&\eta_d
\end{pmatrix}.
\end{align}

Then, $M$ is a density matrix and we call it the QSD-state of $\{\ket{\psi_i},\eta_i\}^d_{i=1}$.

\begin{thm}\cite{Xiong2018A}\label{thm5}
Let $\mathcal{H}$ be a $d$-dimensional Hilbert space and $\{\ket{i}\}^d_{i=1}$ be the computable basis, that is, $\ket{i}=(0,...,0,1,0,...,0)^t$, the $i$-th entry is $1$ for each $i$. For $\ket{\psi_i}\in\mathcal{H}$, the minimal error probability to discriminate the collection of linearly independent pure states $\{\ket{\psi_i},\eta_i\}^d_{i=1}$ is equal to the geometric coherence of the corresponding QSD-state $M$, that is,
\begin{align}\label{eq19}
P^{opt}_E(\{\ket{\psi_i},\eta_i\}^d_{i=1})= C_g(M).
\end{align}

\end{thm}

For $1/2$-affinity and the least square measurement, there exists a similar relationship.
If we denote the corresponding QSD-state by $M$, namely, $\nu_i=M_{ii}=\eta_i$, $\ket{\varphi_i}=\nu^{-1/2}_i\sqrt{M}\ket{i}$ for each $i$, then $\langle\varphi_i\ket{\varphi_j}=(\nu_i\nu_j)^{-1/2}\bra{i}M\ket{j}=\langle\psi_i|\psi_j\rangle$, $1\le i,j\le d$. With Lemma 8 in Ref.  \cite{Xiong2018A}, there exists a unitary $V$ such that $\ket{\varphi_i}=V\ket{\psi_i}$ for each $i$.

As a result, the least square measurement for $\{\ket{\psi_i},\eta_i\}^d_{i=1}$ is
\begin{align}
M_i=\eta_i\rho^{-1/2}_{out}\ket{\psi_i}\bra{\psi_i}\rho^{-1/2}_{out}, i=1,...,d,
\end{align}
with $\rho_{out}=\sum_i\eta_i\ket{\psi_i}\bra{\psi_i}$. Since $\sigma_{out}=\sum_i \eta_i\ket{\varphi_i}\bra{\varphi_i}=V\rho_{out}V^{\dagger}$, the LSM for $\{\ket{\varphi_i},\eta_i\}^d_{i=1}$ is
\begin{align}
N_i&=\eta_i\sigma^{-1/2}_{out}\ket{\varphi_i}\bra{\varphi_i}\sigma^{-1/2}_{out}=VM_iV^{\dagger}.
\end{align}

In addition, one has
\begin{align*}
P^{lsm}_E(\{\ket{\psi_i},\eta_i\})=&\sum_i\eta_i\mathrm{tr}(M_i\ket{\psi_i}\bra{\psi_i})\\
=&\sum_i\eta_i\mathrm{tr}(N_i\ket{\varphi_i}\bra{\varphi_i})\\
=&P^{lsm}_S(\{\ket{\varphi_i},\nu_i\})\\
=&C^{(1/2)}_a(M).
\end{align*}

In conclusion, we have the following result.
\begin{thm}\label{thm6}
Let $\mathcal{H}$ be a $d$-dimensional Hilbert space and $\{\ket{i}\}^d_{i=1}$ be the computable basis, that is, $\ket{i}=(0,...,0,1,0,...,0)^t$, the $i$-th entry is $1$ for $i=1,...,d$. For $\ket{\psi_i}\in\mathcal{H},i=1,...,d$, the error probability to discriminate the collection of pure states $\{\ket{\psi_i},\eta_i\}^d_{i=1}$ with least square measurement is equal to 1/2-affinity of coherence of the corresponding QSD-state $M$, that is,
\begin{align}
P^{lsm}_E(\{\ket{\psi_i},\eta_i\}^d_{i=1})= C^{(1/2)}_a(M),
\end{align}
where the incoherent pure states are $\{\ket{i}\}^d_{i=1}$.

\end{thm}

\subsection{Least square measurement and optimal measurement}

First, we recall the following result.

\begin{thm}\cite{Barnum2002,Spehner2014}
\label{thm-lsm1}
Let $\{\rho_i,\mu_i\}^m_{i=1}$ to be an ensemble of $m$ states of a system in an $n$-dimensional Hilbert space $\mathcal{H}$ ($m\le n$), then
\begin{align}\label{}
P^{opt}_S(\{\rho_i,\mu_i\}^m_{i=1})\le \sqrt{P^{lsm}_S(\{\rho_i,\mu_i\}^m_{i=1})}.
\end{align}
\end{thm}

As $P^{opt}_E=1-P^{opt}_S$ is the minimal error probability of QSD, the error probability with LSM is
 \begin{align}\label{eq22}
 P^{lsm}_E=1-P^{lsm}_S\le 1-(P^{opt}_S)^2\le 2P^{opt}_E.
 \end{align}
Therefore, if $P^{opt}_E$ is very close to 0, so is $P^{lsm}_E$. In fact, LSM is very close to the optimal measurement for almost orthogonal states.

As the LSM to discriminate a set of pure states is actually a von Neumann measurement and the result of Theorem \ref{thm5}, one has
\begin{align*}
2C_g(\rho)\ge C^{(1/2)}_a(\rho)\ge C_g(\rho)\ge \widetilde{C}_a(\rho),
\end{align*}
for any $\rho$. The last inequality is due to Theorem \ref{thm-lsm1} above and Theorem 1 of Ref. \cite{Xiong2018A} as follows: 
$C_g(\rho) \ge 1-P_S^{opt}(\{\rho_i,\mu_i\}_{i=1}^m) \ge 1-\sqrt{P_S^{lsm}(\{\rho_i,\mu_i\}_{i=1}^m)} = \tilde{C}_a(\rho)$.
In addition, since $C_g(\rho)\le\frac{C_{l_1}(\rho)}{d-1}$ for any $\rho>0$ (that is, $\rho$ is invertible) \cite{Xiong2018A}, where $C_{l_1}(\rho)$ is $l_1$-norm of coherence defined as $C_{l_1}(\rho):=\sum_{i\ne j}|\bra{i}\rho\ket{j}|$, we have that for any $\rho>0$, the following inequality holds
\begin{align*}
\frac{2}{d-1}C_{l_1}(\rho)\ge 2C_g(\rho)\ge C^{(1/2)}_a(\rho)\ge C_g(\rho)\ge \widetilde{C}_a(\rho).
\end{align*}

On the other hand, we consider the connection between least square measurement and optimal measurement through coherence.

In Ref. \cite{zhanghj2017}, Zhang {\it et al.} gave an upper bound for geometric coherence as
\begin{align}\label{eq12}
C_g(\rho)\le \min\{l_1,l_2\},
\end{align}
where $l_1=1-\max_i\{\rho_{ii}\}$ and $l_2=1-\sum_ib^2_{ii}$ with $b_{ij}$ being the $(i,j)$-th entry of $\sqrt{\rho}$. This is interesting to note that $l_2$ is actually equal to $C^{(1/2)}_a(\rho)$, and moreover, they also show that $l_2$ is tight for the maximally coherent mixed states given by
 \begin{align}
 \rho_m=p\ket{\psi_d}\bra{\psi_d}+\frac{1-p}{d}I_d,
 \end{align}
where $0<p<1$, and $\ket{\psi_d}=\frac{1}{\sqrt{d}}\sum_i\ket{i}$ is the maximally coherent state.

In other words, one has
 \begin{align}\label{eq13}
 C_g(\rho_m)=C^{(1/2)}_a(\rho_m).
 \end{align}

 Combining Theorem \ref{thm5}, Theorem \ref{thm4} and \eqref{eq13}, we recover the following result.

\begin{thm}\cite{Yuen1975,Helstrom1976}\label{thm3}
For the equiprobable quantum state discrimination task $\{\ket{\phi_i},1/d\}^d_{i=1}$ with $\langle\phi_i\ket{\phi_j}=p$ for $i\ne j$, the least square measurement is optimal. Moreover, the maximum successful probability is
\begin{align*}
P^{opt}_S (\{\ket{\phi_i},1/d\}^d_{i=1})=\left[\frac{d-1}{d}\sqrt{1-p}+\frac{1}{d}\sqrt{1-p+dp}\right]^2.
\end{align*}
\end{thm}

\begin{proof}
Note that the QSD-state of the above-mentioned task is $\rho_m$. As
\begin{align*}
&P^{opt}_E (\{\ket{\phi_i},1/d\}^d_{i=1})=C_g(\rho_m)\\
=&C^{(1/2)}_a(\rho_m)=P^{lsm}_E(\{\ket{\phi_i},1/d\}^d_{i=1}),
\end{align*}
then the least square measurement is optimal. The first equality is the result of Theorem \ref{thm5} and the fact that $\{\ket{\phi_i}\}$ is linearly independent. The last equality is due to Theorem \ref{thm6}.
Using the result in Ref. \cite{zhanghj2017},
\begin{align*}
C_g(\rho_m)=1-\left[\frac{d-1}{d}\sqrt{1-p}+\frac{1}{d}\sqrt{1-p+dp}\right]^2,
\end{align*}
the maximum successful probability is
\begin{align*}
P^{opt}_S (\{\ket{\phi_i},1/d\}^d_{i=1})=\left[\frac{d-1}{d}\sqrt{1-p}+\frac{1}{d}\sqrt{1-p+dp}\right]^2,
\end{align*}
and the corresponding optimal measurement is
\begin{align*}
M^{opt}_i=\frac{1}{d}\rho^{-1/2}_{out}\ket{\phi_i}\bra{\phi_i}\rho^{-1/2}_{out},
\end{align*}
where $\rho_{out}=\frac{1}{d}\sum_i\ket{\phi_i}\bra{\phi_i} ~(i=1,...,d,)$.
\end{proof}

\section{When is LSM optimal?}
\label{sec:qsd-lsm5}

Theorem \ref{thm3} indicates that LSM is optimal for the equiprobable case. However, we find that this is not the only case as discussed below.

\subsection{Two pure states case}

Since we have the explicit expressions of geometric coherence and $1/2$-affinity of coherence for single-qubit states, we can derive the condition for LSM being optimal for an ensemble containing two pure states.
Given an ensemble $\{\ket{\psi_i},\eta_i\}^2_{i=1}$, the corresponding QSD-state is a single-qubit state $\rho=\sum_ic_i\sigma_i$.
From Eq. \eqref{eq6}, one has
 \begin{align*}
 [A^{(1/2)}(\rho)]^2=\frac{1}{2}\left(1+\sqrt{1-|{\bf c}|^2}+\frac{c^2_3}{1+\sqrt{1-|{\bf c}|^2}}\right).
 \end{align*}

On the other hand, with fidelity $F(\rho,\sigma):=\mathrm{tr}\sqrt{\sqrt{\sigma}\rho\sqrt{\sigma}}$,
 \begin{align*}
 F(\rho):=\max_{\sigma\in\mathcal{I}}F(\rho,\sigma)=\sqrt{\frac{1}{2}\left(1+\sqrt{1-c^2_1-c^2_2}\right)}.
 \end{align*}

The above expressions reduce to simpler forms when $\rho$ is a pure state ($|\bf{c}|$ = $\sqrt{c_1^2+c_2^2+c_3^2} = 1$). That is, $[A^{(1/2)}(\rho)]^2=\frac{1}{2}(1+c^2_3)$ and $F^2(\rho)=\frac{1}{2}(1+|c_3|)$. Then, $A^{(1/2)}(\rho)=F(\rho)$ if and only if $c_3=0$ or $\pm1$. The same can be shown true for mixed states with some tedious calculation. Hence, the least square measurement is optimal for two pure states case if and only if these states are orthogonal or have equal probabilities.

\subsection{Multiple copy QSD with LSM}

We consider QSD protocol with multiple copies, as the error probability of a QSD task decreases when we have more copies of states.

For the $N$-copy case $\{\ket{\psi_i}^{\otimes N},\eta_i\}^d_{i=1}$, the $(i,j)$-th entry of the corresponding QSD-state is
\begin{align}
\rho^{(N)}_{ij}=\sqrt{\eta_i\eta_j}\langle\psi_i\ket{\psi_j}^N~(1\le i,j\le d).\nonumber
\end{align}

Let $N\rightarrow\infty$ and $\rho^{(N)}_{ij}\rightarrow0$ for each $i\ne j$. Since $\{\ket{\psi_i}^{\otimes N}\}^d_{i=1}$ is linearly independent for large $N$, the QSD-state $\rho^{(N)}$ is invertible. Then,
\begin{align*}
C^{(1/2)}_a(\rho)\le \frac{2}{d-1}C_{l_1}(\rho),
\end{align*}
and the error probability to discriminate $\{\ket{\psi_i}^{\otimes N},\eta_i\}^d_{i=1}$ tends to zero. In other words, if we have enough copies of states, pure states $\{\ket{\psi_i},\eta_i\}^d_{i=1}$ can be almost perfectly distinguished by the LSM. In other words, we prove that LSM is asymptotically optimal for discrimination of pure states in the sense that the corresponding QSD-state $\rho\rightarrow\rho^{diag}=\sum_i\bra{i}\rho\ket{i}\ket{i}\bra{i}$.

\section{Duality between 1/2-affinity of coherence and path distinguishability}
\label{sec:duality6}

Wave-particle duality is an intriguing but central concept in quantum physics. In double-slit interference experiment, a single quantum object can exhibit the wave nature as long as knowledge about the path chosen by the object is uncertain. More knowledge about the path corresponds to poor interference. In this direction, quantitative relations in the form of trade-off between wave and particle aspects were studied by Greenberger-Yasin \cite{greenberger-yasin} and Englert \cite{BGEnglert}, respectively. Englert, in his famous paper \cite{BGEnglert}, derived a path-visibility duality relation for the optimal detector measurement for two paths as follows:
\begin{equation}
\mathcal{V}^2 + \mathcal{D}^2 \leq 1,
\end{equation}
where $\mathcal{V}$ is the visibility of the interference pattern and $\mathcal{D}$ is a measure of path distinguishability or which-way information.
Recently, Bera {\it et al.} \cite{Bera2015} obtained a complementarity relation between $l_1$- norm of coherence and path distinguishability in the case of Yang's $n$-slit experiment. Here, although we are unable to provide a general proof for mixed states in arbitrary dimensions, we establish the complementarity between 1/2-affinity of coherence and path distinguishability for some special cases.

Consider the case of $d$-slit quantum interference with pure quantons. In the Yang's $n$-slit experiment, if the quanton passes through the $i$th slit or takes the $i$th path, then we denote $\ket{i}$ as the possible state. As a result, the state of the quanton can be represented with $d$ basis states $\{\ket{1},...,\ket{d}\}$ as
\begin{align}
\ket{\Psi}=c_1\ket{1}+...+c_d\ket{d},
\end{align}
where $\ket{i}$ represents the $i$th slit and $c_i$ is the amplitude of taking the $i$th slit. To determine through which slit the quanton passes, one needs to perform a quantum measurement. According to quantum measurement theory, the quanton will interact with a detector state and the compound state is given by
\begin{align}
U(\ket{\Psi}\ket{0_d})=\sum_ic_i\ket{i}\ket{d_i},
\end{align}
where $\{\ket{d_i}\}$ are normalized but not necessarily orthogonal states of the detector.

To quantify the coherence of quanton, one considers the reduced density matrix of the quanton after tracing out the detector states,
\begin{align}
\rho_s=\sum^d_{i,j=1}c_i\bar{c}_j\langle d_j\ket{d_i}\ket{i}\bra{j}.
\end{align}

From Theorem \ref{thm1}, the $1/2$-affinity of coherence is
\begin{align}
C^{(1/2)}_a(\rho_s)=1-P^{lsm}_S(\{\ket{\psi_i},\eta_i\}^d_{i=1}),
\end{align}
where $\eta_i=|c_i|^2,\ket{\psi_i}=\exp(\sqrt{-1}\theta_i)\eta^{-1/2}_i\sqrt{\rho_s}\ket{i}$ and $\theta_i$ is the argument of $c_i$.

Now, to know which path the quanton takes, one has to discriminate the detector states $\{\ket{d_i},|c_i|^2\}^d_{i=1}$. In other words, the path distinguishability is actually equivalent to the discrimination of the corresponding detector states.

Since $\langle \psi_i\ket{\psi_j}=\langle d_j\ket{d_i}=\langle \overline{d_i}\ket{\overline{d_j}}$, there exists a unitary matrix $V$ such that $\ket{d_i}=V\ket{\overline{\psi_i}}$.
Therefore, one has
\begin{align*}
\rho_{out}=\sum_i|c_i|^2\ket{d_i}\bra{d_i}=V\sum_i|c_i|^2\ket{\overline{\psi_i}}\bra{\overline{\psi_i}}V^{\dagger}=V\overline{\rho_s}V^{\dagger},
\end{align*}
and then the corresponding LSM for $\{\ket{d_i},|c_i|^2\}$ is
\begin{align*}
N^{lsm}_i=|c_i|^2\rho^{-1/2}_{out}\ket{d_i}\bra{d_i}\rho^{-1/2}_{out}=V\ket{i}\bra{i}V^{\dagger}.
\end{align*}

As a result, one has
\begin{align*}
&P^{lsm}_S \left(\{\ket{d_i},|c_i|^2\}^d_{i=1}\right)=\sum_i|c_i|^2|\bra{i}V^{\dagger}\ket{d_i}|^2\\
=&\sum_i|\bra{i}\sqrt{\overline{\rho_s}}\ket{i}|^2=P^{lsm}_S(\{\ket{\psi_i},|c_i|^2\}^d_{i=1}).
\end{align*}

Even though it is not the optimal choice for quantum state discrimination, LSM is very close to the optimal one when the states to be distinguished are almost orthogonal, and its construction is also relatively simple. Moreover, the complementarity between coherence and path distinguishability holds just for linearly independent detector states \cite{Bera2015,Xiong2018A}. Therefore, if we define the optimal successful probability to discriminate the detector states with LSM as path distinguishability, $D_q:=P^{lsm}_S(\{\ket{d_i},|c_i|^2\}^d_{i=1})$,
and the $1/2$-affinity of coherence as coherence, $C:=C^{(1/2)}_a(\rho_s)$,
we obtain the complementarity between $1/2$-affinity of coherence and path distinguishability as
\begin{align}\label{eq21}
C+D_q=1.
\end{align}

Thus, the wave nature of the quanton can also be characterized by $C^{(1/2)}_a(\rho_s)$. If the quantum system is exposed to the environment, that is, the quanton state is a mixed state  $\rho=\sum_{i,j}\rho_{ij}\ket{i}\bra{j}$, we can obtain a generalized complementarity. The composite system of the quanton and the path detector after the unitary interaction can be given as
\begin{align}
\rho_{sd}=\sum_{i,j}\rho_{ij}\ket{i}\bra{j}\otimes\ket{d_i}\bra{d_j},
\end{align}
and the reduced density matrix of the quanton after tracing out the detector states is
\begin{align}\label{eq24}
\rho_s=\sum^d_{i,j=1}\rho_{ij}\langle d_j\ket{d_i}\ket{i}\bra{j}.
\end{align}

As every principal $2\times2$ submatrix in Eq. \eqref{eq24} is positive semidefinite \cite[p.434]{Horn:2012:MA:2422911}, we have
\begin{align}
\sqrt{\rho_{ii}\rho_{jj}}-|\rho_{ij}|\ge0 ~(1\le i,j\le d).
\end{align}
Assuming that the corresponding ensemble to $\rho_s$ is $\{\ket{\psi_i},\rho_{ii}\}$, we have
\begin{align*}
|\langle\psi_i\ket{\psi_j}|=\frac{|\bra{i}\rho_s\ket{j}|}{\sqrt{\rho_{ii}\rho_{jj}}}=\frac{|\rho_{ij}|}{\sqrt{\rho_{ii}\rho_{jj}}}|\langle d_i\ket{d_j}|\le|\langle d_i\ket{d_j}|,
\end{align*}
for each $i$ and $j$. In other words, a pair of states in $\{\ket{d_i},\rho_{ii}\}^d_{i=1}$ is more difficult to distinguish than the corresponding pair in $\{\ket{\psi_i},\rho_{ii}\}^d_{i=1}$. 
%
%
Suppose $\rho^{\prime}_s=\sum^d_{i,j=1}\sqrt{\rho_{ii}\rho_{jj}}\langle d_j\ket{d_i}\ket{i}\bra{j}$. In the cases where there exists an incoherent operation $\Phi$ (see Appendix \ref{app4}), we have
\begin{align}
\Phi(\rho^{\prime}_s)=\rho_s.
\end{align}

As $C^{(1/2)}_a$ is a coherence measure, we have
\begin{align*}
1-C^{(1/2)}_a(\rho_s)&\ge 1-C^{(1/2)}_a(\rho^{\prime}_s) \nonumber \\
&=P^{lsm}_S(\{\ket{d_i},\rho_{ii}\}^d_{i=1}) \equiv D_q.
\end{align*}

Hence, we have the following complementarity relation between coherence and path distinguishability,
\begin{align}
C+D_q\le1.
\end{align}

\section{Conclusion}
\label{sec:conclusion7}

In this paper, we have introduced a family of coherence measures, namely $\alpha$-affinity of coherence for $\alpha\in(0,1)$. Moreover, we obtained the analytic formulae for these quantifiers and also studied their convex roof extension. In particular, we have offered an operational meaning for $1/2$-affinity of coherence, by showing that this equals the error probability to discriminate a set of pure states with least square measurement. Based on the relationship between the LSM and the optimal measurement, we obtained the optimal measurement for the equiprobable quantum state discrimination. Furthermore, we obtained conditions for the LSM to be the optimal measurement for two pure states from the perspective of coherence theory. In addition, we also studied the multiple copy QSD and concluded that LSM is optimal in the asymptotical sense. At last, we established the complementary relationship between $1/2$-affinity of coherence and path distinguishability.

Our results not only offer a class of bona fide coherence quantifiers, but also reveal a close link between the quantification of coherence and quantum state discrimination. However, the operational interpretation of general $\alpha$-affinity coherence needs further investigation. \\

\begin{note}
After the completion of this work we were informed by Hyukjoon Kwon that $1/2$-affinity of coherence has been computed and proven to be a coherence measure independently in Refs. \cite{Yucs2017,Kwon2018} by different methods, yielding the same result.
\end{note}

\begin{acknowledgments}
Authors thank Jon Tyson and the anonymous Referee(s) for their useful comments. This project is supported by National Natural Science Foundation of China (Grants No.11171301 and No. 11571307).
\end{acknowledgments}

\appendix

\section{$A^{(\alpha)}$ is bounded}\label{app3}
\begin{prop}
 $0\le A^{(\alpha)}(\rho,\sigma)\le1$, with $A^{(\alpha)}(\rho,\sigma)=1$ if and only if $\rho=\sigma$.
\end{prop}

\begin{proof}
As $\rho^{\frac{\alpha}{2}}\sigma^{1-\alpha}\rho^{\frac{\alpha}{2}}$ is a positive matrix, one has
\begin{align}
\mathrm{Tr}(\rho^{\alpha}\sigma^{1-\alpha})=\mathrm{Tr}(\rho^{\frac{\alpha}{2}}\sigma^{1-\alpha}\rho^{\frac{\alpha}{2}})\ge0\nonumber.
\end{align}

The other part can be proved as in Ref. \cite{Wilde2014}. Let $\{\ket{x}\}_x$ be a basis of $\mathcal{H}$, then $M=\{M_x~|~M_x=\ket{x}\bra{x}\}$ is an informationally-complete measurement. Denoting $\Phi(\rho)=\sum_x\bra{x}\rho\ket{x}\ket{x}\bra{x}$, we have from the monotonicity of $A^{(\alpha)}(\rho,\sigma)$ and Jensen's inequality
\begin{align}
A^{(\alpha)}(\rho,\sigma)&\le A^{(\alpha)}(\Phi(\rho),\Phi(\sigma))\nonumber\\
&=\sum_x \left(\frac{\bra{x}\rho\ket{x}}{\bra{x}\sigma\ket{x}}\right)^{\alpha} \bra{x}\sigma\ket{x}\nonumber\\
&\le \left(\sum_x\bra{x}\rho\ket{x}\right)^{\alpha}=1.\nonumber
\end{align}

As the equality holds iff $\bra{x}\rho\ket{x}=\bra{x}\sigma\ket{x}$ for any informationally-complete measurement, one has $A^{(\alpha)}(\rho,\sigma)=1$ if and only if $\rho=\sigma$.
\end{proof}

\begin{widetext}
\section{$C^{(1/2)}_a(\rho_s)\le C^{(1/2)}_a(\rho^{\prime}_s)$}\label{app4}

\subsection{d=2 case}

For $d=2$, let $\Phi=\{K_{12},K_{11},K_{22}\}$ with
\begin{align}
K_{12}=\begin{pmatrix}
\frac{\sqrt{\rho_{12}}}{(\rho_{11}\rho_{22})^{1/4}}&0\\
0&\frac{\sqrt{\rho_{21}}}{(\rho_{11}\rho_{22})^{1/4}}\\
\end{pmatrix},\nonumber
~K_{11}=\begin{pmatrix}
\sqrt{1-\frac{|\rho_{12}|}{\sqrt{\rho_{11}\rho_{22}}}}&0\\
0&0\\
\end{pmatrix},\nonumber
~K_{22}=\begin{pmatrix}
0&0\\
0&\sqrt{1-\frac{|\rho_{12}|}{\sqrt{\rho_{11}\rho_{22}}}}\\
\end{pmatrix}.\nonumber
\end{align}

Since $\frac{|\rho_{12}|}{\sqrt{\rho_{11}\rho_{22}}}\le1$ we have $\left|\frac{\sqrt{\rho_{12}}}{(\rho_{11}\rho_{22})^{1/4}} \right|^2+\left|\sqrt{1-\frac{|\rho_{12}|}{\sqrt{\rho_{11}\rho_{22}}}} \right|^2=1$, and $\Phi$ is an incoherent operation such that $\Phi(\rho^{\prime}_s)=\rho_s$. Hence,
$C^{(1/2)}_a(\rho^{\prime}_s)\ge C^{(1/2)}_a(\rho_s)$.

\subsection{d=3 case}
For $d=3$, we denote $\sigma_{ij}=\frac{\rho_{ij}}{\sqrt{\rho_{ii}\rho_{jj}}}$ and $\rho_{ij}=|\rho_{ij}|e^{i\theta_{ij}}$. Without any loss of generality, we can assume that $|\sigma_{12}|\ge|\sigma_{13}|\ge|\sigma_{23}|$. Then the quantum operation $\Phi=\{K_{12},K_{13},K_{11},K_{22},K_{33}\}$ is 
\begin{align}
K_{12}=\begin{pmatrix}
\sqrt{\sigma_{12}}&0&0\\
0&\sqrt{\overline{\sigma_{12}}}&0\\
0&0&\frac{\overline{\sigma_{23}}}{\sqrt{\sigma_{12}}}
\end{pmatrix},\nonumber
~K_{13}=\begin{pmatrix}
\sqrt{\sigma_{13}-\sigma_{23}e^{i\theta_{12}}}&0&0\\
0&0&0\\
0&0&\sqrt{\sigma_{31}-\sigma_{32}e^{-i\theta_{12}}}
\end{pmatrix},\nonumber
\end{align}

\begin{align}
K_{11}=\begin{pmatrix}
\sqrt{1-|\sigma_{12}|-|\sigma_{13}-\sigma_{23}e^{i\theta_{12}}|}&0&0\\
0&0&0\\
0&0&0
\end{pmatrix},\nonumber
~K_{22}=\begin{pmatrix}
0&0&0\\
0&\sqrt{1-|\sigma_{12}|}&0\\
0&0&0
\end{pmatrix},\nonumber
~K_{33}=\begin{pmatrix}
0&0&0\\
0&0&0\\
0&0&\sqrt{1-\frac{|\sigma_{23}|^2}{|\sigma_{12}|}-|\sigma_{13}-\sigma_{23}e^{i\theta_{12}}|}
\end{pmatrix}.\nonumber
\end{align}

If $|\sigma_{12}|+|\sigma_{13}-\sigma_{23}e^{i\theta_{12}}|\le1$ and hence $|\sigma_{13}-\sigma_{23}e^{i\theta_{12}}|\le1$, then $\Phi$ is an incoherent operation. But, this may not be true for all states as these conditions may not be satisfied. Moreover, 
\begin{align}
K_{12}\rho^{\prime}_sK^{\dagger}_{12}=\begin{pmatrix}
|\sigma_{12}|\rho_{11}&\rho_{12}\langle{d_2}\ket{d_1}&\frac{\sqrt{\rho_{11}}\rho_{23}e^{i\theta_{12}}}{\sqrt{\rho_{22}}}\langle{d_3}\ket{d_1}\\
\rho_{21}\langle{d_1}\ket{d_2}&|\sigma_{12}|\rho_{22}&\rho_{23}\langle{d_3}\ket{d_2}\\
\frac{\sqrt{\rho_{11}}\rho_{32}e^{-i\theta_{12}}}{\sqrt{\rho_{22}}}\langle{d_1}\ket{d_3}&\rho_{32}\langle{d_2}\ket{d_3}&\frac{|\sigma_{23}|^2}{|\sigma_{12}|}\rho_{33}
\end{pmatrix}\nonumber
\end{align}

and
\begin{align}
K_{13}\rho^{\prime}_sK^{\dagger}_{13}=\begin{pmatrix}
|\sigma_{13}-\sigma_{23}e^{i\theta_{12}}|\rho_{11}&0&(\rho_{13}-\frac{\sqrt{\rho_{11}}\rho_{23}e^{i\theta_{12}}}{\sqrt{\rho_{22}}})\langle{d_3}\ket{d_1}\\
0&0&0\\
(\rho_{31}-\frac{\sqrt{\rho_{11}}\rho_{32}e^{-i\theta_{12}}}{\sqrt{\rho_{22}}})\langle{d_3}\ket{d_1}&0&|\sigma_{13}-\sigma_{23}e^{i\theta_{12}}|\rho_{33}
\end{pmatrix}.\nonumber
\end{align}

As a result, $\Phi(\rho^{\prime}_s)=\rho_s$ and $C^{(1/2)}_a(\rho_s)\le C^{(1/2)}_a(\rho^{\prime}_s)$.

\subsection{finite dimensional case}

If $\sum_{j\ne i}\frac{|\rho_{ij}|}{\sqrt{\rho_{ii}\rho_{jj}}}\le1$ (for each $i,j$), then the duality relation is true.
We denote the Kraus operators of quantum operation $\Phi \equiv \{K_{ij}\}$ ($1 \le i \le j \le d$) as follows:

\begin{align}
K_{ij}(i< j)=\begin{pmatrix}
0&...&0&...&0&...&0\\
.&...&.&...&.&...&.\\
.&...&.&...&.&...&.\\
.&...&.&...&.&...&.\\
0&...&\frac{\sqrt{\rho_{ij}}}{(\rho_{ii}\rho_{jj})^{1/4}}&...&0&...&0\\
.&...&.&...&.&...&.\\
.&...&.&...&.&...&.\\
.&...&.&...&.&...&.\\
0&...&0&...&\frac{\sqrt{\rho_{ji}}}{(\rho_{ii}\rho_{jj})^{1/4}}&...&0\\
.&...&.&...&.&...&.\\
.&...&.&...&.&...&.\\
.&...&.&...&.&...&.\\
0&...&0&...&0&...&0
\end{pmatrix}_{d \times d},\nonumber
~K_{ii}=\begin{pmatrix}
0&...&0&...&0&...&0\\
.&...&.&...&.&...&.\\
.&...&.&...&.&...&.\\
.&...&.&...&.&...&.\\
0&...&\sqrt{1-\sum_{j\ne i}\frac{|\rho_{ij}|}{\sqrt{\rho_{ii}\rho_{jj}}}}&...&0&...&0\\
.&...&.&...&.&...&.\\
.&...&.&...&.&...&.\\
.&...&.&...&.&...&.\\
0&...&0&...&0&...&0
\end{pmatrix}_{d \times d}.\nonumber
\end{align}

Then, it is not difficult to check that $\Phi$ is an incoherent operation and, for $\rho^{\prime}_s=\sum^d_{i,j=1}\sqrt{\rho_{ii}\rho_{jj}}\langle d_j\ket{d_i}\ket{i}\bra{j}$, we have $\Phi(\rho^{\prime}_s)=\sum^d_{i,j=1}\rho_{ij}\langle d_j\ket{d_i}\ket{i}\bra{j}=\rho_s$. Since $C^{(1/2)}_a$ is a coherence measure, we have
\begin{align}
C^{(1/2)}_a(\rho_s)=C^{(1/2)}_a[\Phi(\rho^{\prime}_s)]\le C^{(1/2)}_a(\rho^{\prime}_s).
\end{align}
%
\end{widetext}


\begin{thebibliography}{99}%

\bibitem{Baumgratz2014} T. Baumgratz, M. Cramer, and M. B. Plenio, {\it Quantifying coherence}, Phys. Rev. Lett. {\bf 113}, 140401 (2014).

\bibitem{Adesso2016BC} G. Adesso, T. R. Bromley, and M. Cianciaruso, {\it Measures and applications of quantum correlations}, J. Phys. A: Math. Theor. {\bf 49}, 473001 (2016).

\bibitem{Yao2015} Y. Yao, X. Xiao, L. Ge, and C. P. Sun, {\it Quantum coherence in multipartite systems}, Phys. Rev. A {\bf 92}, 022112 (2015).

\bibitem{Matera2016A} J. M. Matera, D. Egloff, N. Killoran, and M. B. Plenio, {\it Coherent control of quantum systems as a resource theory}, Quant. Sci. Tech. {\bf 1}, 01LT01 (2016).

\bibitem{Hillery2016A} M. Hillery, {\it Coherence as a resource in decision problems: The Deutsch-Jozsa algorithm and a variation}, Phys. Rev. A {\bf 93}, 012111 (2016).

\bibitem{Anand2017} N. Anand and A. K. Pati, {\it Coherence and entanglement monogamy in the discrete analogue of analog Grover search}, arXiv:1611.04542.

\bibitem{Shi2017} H.-L. Shi, S.-Y. Liu, X.-H. Wang, W.-L. Yang, Z.-Y. Yang, and H. Fan, {\it Coherence depletion in the Grover quantum search algorithm}, Phys. Rev. A {\bf 95}, 032307 (2017).

\bibitem{Girolami14} D. Girolami, {\it Observable measure of quantum coherence in finite dimensional systems}, Phys. Rev. Lett. {\bf 113}, 170401 (2014).

\bibitem{Streltsov2015B} A. Streltsov, U. Singh, H. S. Dhar, M. N. Bera, and G. Adesso, {\it Measuring quantum coherence with entanglement}, Phys. Rev. Lett. {\bf 115}, 020403 (2015).

\bibitem{Chitambar2016A} E. Chitambar and G. Gour, {\it Critical examination of incoherent operations and a physically consistent resource theory of quantum coherence}, Phys. Rev. Lett. {\bf 117}, 030401 (2016).

\bibitem{chitambar2016B} E. Chitambar and G. Gour, {\it Comparison of incoherent operations and measures of coherence}, Phys. Rev. A {\bf 94}, 052336 (2016).

\bibitem{streltsov2017A} A. Streltsov, G. Adesso, and M. B. Plenio, {\it Colloquium: Quantum coherence as a resource}, Rev. Mod. Phys. {\bf 89}, 041003 (2017).

\bibitem{Yuan2015} X. Yuan, H. Zhou, Z. Cao, and X. Ma, {\it Intrinsic randomness as a measure of quantum coherence}, Phys. Rev. A {\bf 92}, 022124 (2015).

\bibitem{winter2016} A. Winter and D. Yang, {\it Operational resource theory of coherence}, Phys. Rev. Lett. {\bf 116}, 120404 (2016).

\bibitem{Napoli2016} C. Napoli, T. R. Bromley, M. Cianciaruso, M. Piani, N. Johnston, and G. Adesso, {\it Robustness of coherence: An operational and observable measure of quantum coherence}, Phys. Rev. Lett. {\bf 116}, 150502 (2016).

\bibitem{Bu2017A} K. Bu, U. Singh, S.-M. Fei, A. K. Pati, and J. Wu, {\it Maximum relative entropy of coherence: An operational coherence measure}, Phys. Rev. Lett. {\bf 119}, 150405 (2017).

\bibitem{HorodeckiRMP09} R. Horodecki, P. Horodecki, M. Horodecki, and K. Horodecki, {\it Quantum entanglement}, Rev. Mod. Phys. {\bf 81}, 865 (2009).

\bibitem{Vedral1997} V. Vedral, M. B. Plenio, M. A. Rippin, and P. L. Knight, {\it Quantifying entanglement}, Phys. Rev. Lett. {\bf 78}, 2275 (1997).

\bibitem{Vedral1998} V. Vedral and M. B. Plenio, {\it Entanglement measures and purification procedures}, Phys. Rev. A {\bf 57}, 1619 (1998).

\bibitem{Plenio2000} M. B. Plenio, S. Virmani, and P. Papadopoulos, {\it Operator monotones, the reduction criterion and the relative entropy}, J. Phys. A: Math. Gen. {\bf 33}, L193 (2000).

\bibitem{X.yu2016B} X.-D. Yu, D.-J. Zhang, G. F. Xu, and D. M. Tong, {\it Alternative framework for quantifying coherence}, Phys. Rev. A {\bf 94}, 060302 (2016).

\bibitem{Helstrom1976} C. W. Helstrom, {\it Quantum detection and estimation theory} (New York: Academic Press, 1976).

\bibitem{Holevo2011} A. S. Holevo, {\it Probabilistic and statistical aspects of quantum theory} (Edizioni della Normale, 2011).

\bibitem{Holevo2001} A. S. Holevo, {\it Statistical structure of quantum theory} (Springer, Dordrecht, 2001).

\bibitem{HELSTROM1967A} C. W. Helstrom, {\it Detection theory and quantum mechanics}, Inf. Cont. {\bf 10}, 254 (1967).

\bibitem{HELSTROM1968} C. W. Helstrom, Inf. Cont. {\bf 13}, 156 (1968).

\bibitem{HOLEVO1973337} A. S. Holevo, {\it Statistical decision theory for quantum systems}, J. Multivar. Anal. {\bf 3}, 337 (1973).

\bibitem{Yuen1975} H. Yuen, R. Kennedy, and M. Lax, {\it Optimum testing of multiple hypotheses in quantum detection theory}, IEEE Trans. Inf. Theory {\bf 21}, 125 (1975). 

\bibitem{Davies1978} E. B. Davies, {\it Information and quantum measurement}, IEEE Trans. Inf. Theory IT-{\bf 24}, 596 (1978).

\bibitem{Phoenix1995} S. J. D. Phoenix and P. D. Townsend, {\it Quantum cryptography: How to beat the code breakers using quantum mechanics}, Contemp. Phys. {\bf 36}, 165 (1995).

\bibitem{HKLo2011} H.-K. Lo, S. Popescu, and T. P. Spiller, {\it Introduction to Quantum Computation and Information} ((Singapore: World-Scientific, 1998).

\bibitem{Bouwmeester2000} D. Bouwmeester, A. K. Ekert, and A. Zeilinger, {\it The physics of quantum information: quantum cryptography, quantum teleportation, quantum computation} (Springer-Verlag, Berlin, Heidelberg, 2000).

\bibitem{Gisin2002RMP} N. Gisin, G. Ribordy, W. Tittel, and H. Zbinden, {\it Quantum cryptography}, Rev. Mod. Phys. {\bf 74}, 145 (2002).

\bibitem{Loepp2006} W. K. W. Susan Loepp, {\it Protecting Information From Classical Error Correction to Quantum Cryptography} (Cambridge University Press, 2006).

\bibitem{Belavkin1975a} V. P. Belavkin, {\it Optimum distinction of nonorthogonal quantum signals}, Radiotekhnika i Elektronika {\bf 20}, 1177 (1975).

\bibitem{Belavkin1975} V. P. Belavkin, {\it Optimal quantum multiple hypothesis testing}, Stochastics {\bf 1}, 315 (1975). 

\bibitem{Holevo1978} A. S. Holevo, Teoriya Veroyatnostej i Ee Primeneniya, 23 (1978). 

\bibitem{Hausladen1994} P. Hausladen and W. K. Wootters, {\it A `pretty good' measurement for distinguishing quantum states}, J. Mod. Opt. {\bf 41}, 2385 (1994).

\bibitem{Hausladen1996} P. Hausladen, R. Jozsa, B. Schumacher, M. Westmoreland, and W. K. Wootters, {\it Classical information capacity of a quantum channel}, Phys. Rev. A {\bf 54}, 1869 (1996).

\bibitem{peres1991} A. Peres and W. K. Wootters, {\it Optimal detection of quantum information}, Phys. Rev. Lett. {\bf 66}, 1119 (1991). 

\bibitem{Eldar2001} Y. C. Eldar and G. D. Forney, {\it On quantum detection and the square-root measurement}, IEEE Trans. Inf. Theory {\bf 47}, 858 (2001).

\bibitem{Spehner2014} D. Spehner, {\it Quantum correlations and distinguishability of quantum states}, J. Math. Phys. {\bf 55}, 075211 (2014).

\bibitem{note} For a pure state ensemble $\{\ket{\psi_i},\eta_i\}$, the LSM is given by $\{\ket{\mu_i}\bra{\mu_i}\}$, where $\ket{\mu_i}=\sqrt{\eta_i}(\sum_j\eta_j\ket{\psi_i}\bra{\psi_i})^{-1/2}\ket{\psi_i}$ satisfies the measurement minimum $\sum_i||\ket{\mu_i}-\sqrt{\eta_i}\ket{\psi_i}||^2$ under the constraint $\sum_i\ket{\mu_i}\bra{\mu_i}=I$.

\bibitem{LeCam1986} L. M. Le Cam, {\it Asymptotic methods in statistical theory} (Springer-Verlag, Berlin, Heidelberg, 1986).

\bibitem{Bhattacharyya-measure} A. Bhattacharyya, {\it On a measure of divergence between two statistical populations defined by their probability distribution}, Bulletin of the Calcutta Mathematical Society {\bf 35}, 99 (1943).

\bibitem{Luo2004} S. Luo and Q. Zhang, {\it Informational distance on quantum-state space}, Phys. Rev. A {\bf 69}, 032106 (2004).

\bibitem{Nielsen10} M. A. Nielsen and I. L. Chuang, {\it Quantum computation and quantum information} (Cambridge University Press, 2010).

\bibitem{Audenaert2015} K. M. R. Audenaert and N. Datta, {\it $\alpha$-z-R{\'e}nyi relative entropies}, J. Math. Phys. {\bf 56}, 022202 (2015). 

\bibitem{LIEB1973} E. H. Lieb, {\it Convex trace functions and the Wigner-Yanase-Dyson conjecture}, Adv. Math. {\bf 11}, 267 (1973).

\bibitem{Audenaert2007} K. M. R. Audenaert, {\it A sharp continuity estimate for the von Neumann entropy}, J. Phys. A: Math. Theor. {\bf 40}, 8127 (2007).

\bibitem{Audenaert2007A} K. M. R. Audenaert, J. Calsamiglia, R. Mun\~{o}z Tapia, E. Bagan, L. Masanes, A. Acin, and F. Verstraete, {\it Discriminating states: The quantum Chernoff bound}, Phys. Rev. Lett. {\bf 98},160501 (2007).

\bibitem{spehner2013A} D. Spehner and M. Orszag, {\it Geometric quantum discord with Bures distance}, New J. Phys. {\bf 15}, 103001 (2013).

\bibitem{spehner2013B} D. Spehner and M. Orszag, {\it Geometric quantum discord with Bures distance: the qubit case}, J. Phys. A: Math. Theor. {\bf 47}, 035302 (2014).

\bibitem{Roga2016} W. Roga, D. Spehner, and F. Illuminati, {\it Geometric measures of quantum correlations}, J. Phys. A: Math. Theor. {\bf 49}, 235301 (2016).

\bibitem{trivial-coherence-measure} K. C. Tan, S. Choi, H. Kwon, and H. Jeong, {\it Coherence, quantum Fisher information, superradiance, and entanglement as interconvertible resources}, Phys. Rev. A {\bf 97}, 052304 (2018).


\bibitem{Spehner2017A} D. Spehner, F. Illuminati, M. Orszag, and W. Roga, {\it Geometric measures of quantum correlations with Bures and Hellinger distances}, in {\it Lectures on General Quantum Correlations and their Applications}, edited by F. F. Fanchini, D. d. O. Soares Pinto, and G. Adesso (Springer International Publishing, Cham, 2017) pp. 105–157.

\bibitem{Xiong2018A} C. Xiong and J. Wu, {\it Geometric coherence and quantum state discrimination}, arXiv:1801.06031v3 ({\it to be published in J. Phys. A: Math. Theor.}).

\bibitem{Barnum2002} H. Barnum and E. Knill, {\it Reversing quantum dynamics with near-optimal quantum and classical fidelity}, J. Math. Phys. {\bf 43}, 2097 (2002).

\bibitem{zhanghj2017} H.-J. Zhang, B. Chen, M. Li, S.-M. Fei, and G.-L. Long, {\it Estimation on geometric measure of quantum coherence}, Commun. Theor. Phys. {\bf 67}, 166 (2017).


\bibitem{greenberger-yasin} D. M. Greenberger and A. Yasin, {\it Simultaneous wave and particle knowledge in a neutron interferometer}, Phys. Lett. A {\bf 128}, 391 (1988).

\bibitem{BGEnglert} B-G. Englert, {\it Fringe visibility and which-way information: An inequality}, Phys. Rev. Lett. {\bf 77}, 2154 (1996).

\bibitem{Bera2015} M. N. Bera, T. Qureshi, M. A. Siddiqui, and A. K. Pati, {\it Duality of quantum coherence and path distinguishability}, Phys. Rev. A {\bf 92}, 012118 (2015).

\bibitem{Horn:2012:MA:2422911} R. A. Horn and C. R. Johnson, {\it Matrix analysis}, (Cambridge University Press, New York, USA, 2012).

\bibitem{Yucs2017} C.-s. Yu, {\it Quantum coherence via skew information and its polygamy}, Phys. Rev. A {\bf 95}, 042337 (2017).

\bibitem{Kwon2018} H. Kwon, C.-Y. Park, K. C. Tan, D. Ahn, and H. Jeong, {\it Coherence, asymmetry, and quantum macroscopicity}, Phys. Rev. A {\bf 97}, 012326 (2018).

\bibitem{Wilde2014} M. M. Wilde, A. Winter, and D. Yang, {\it Strong converse for the classical capacity of entanglement-breaking and Hadamard channels via a sandwiched R{\'e}nyi relative entropy}, Commun. Math. Phys. {\bf 331}, 593 (2014).

\end{thebibliography}

%

\end{document}